\documentclass[11pt]{article}
\usepackage[margin=1.1in]{geometry}
\usepackage[round]{natbib}
\usepackage{amsmath, amssymb, enumitem, bm, bbm, multirow, hyperref, graphicx, authblk, tikz, amsthm, listings}
\usepackage{booktabs}
\usepackage{threeparttable}

\usepackage[noabbrev]{cleveref}
\usepackage{thmtools,thm-restate}

\makeatletter
\newcommand{\refcheckize}[1]{%
  \expandafter\let\csname @@\string#1\endcsname#1%
  \expandafter\DeclareRobustCommand\csname relax\string#1\endcsname[1]{%
    \csname @@\string#1\endcsname{##1}\@for\@temp:=##1\do{\wrtusdrf{\@temp}\wrtusdrf{{\@temp}}}}%
  \expandafter\let\expandafter#1\csname relax\string#1\endcsname
}
\newcommand{\refcheckizetwo}[1]{%
  \expandafter\let\csname @@\string#1\endcsname#1%
  \expandafter\DeclareRobustCommand\csname relax\string#1\endcsname[2]{%
    \csname @@\string#1\endcsname{##1}{##2}\wrtusdrf{##1}\wrtusdrf{{##1}}\wrtusdrf{##2}\wrtusdrf{{##2}}}%
  \expandafter\let\expandafter#1\csname relax\string#1\endcsname
}
\makeatother

\hypersetup{colorlinks, breaklinks,urlcolor= magenta, linkcolor= blue, citecolor = magenta}

\newcommand{\E}{\mathbb{E}}

\newcommand{\Var}{\text{Var}}
\newcommand{\Cov}{\text{Cov}}

\newcommand{\indep}{\perp \!\!\! \perp}
\allowdisplaybreaks
\makeatletter
\renewcommand\AB@affilsepx{, \protect\Affilfont}
\makeatother

\title{The Role of Measured Covariates in Assessing Sensitivity to Unmeasured Confounding}
\author[1]{Abhinandan Dalal*}
\author[1]{Iris Horng*}
\author[2]{Yang Feng}
\author[1]{Dylan S. Small}
\affil[1]{University of Pennsylvania}
\affil[2]{New York University}

\date{}

\newtheorem{assumption}{Assumption}[section]

\newtheorem{proposition}{Proposition}

\newcommand{\xs}{\fontsize{8}{9}\selectfont}
\lstdefinestyle{rout}{
  basicstyle=\ttfamily\xs,
  columns=fixed,
  keepspaces=true,
  breaklines=false,
  breakatwhitespace=false,
  showstringspaces=false,
  frame=single,
  framerule=0.3pt,
  rulecolor=\color{black!30},
  backgroundcolor=\color{black!3},
  xleftmargin=0pt,
  framexleftmargin=4pt,
  framesep=3pt,
  basewidth=0.5em
}

\begin{document}
\maketitle

\begin{abstract}
    Sensitivity analysis is widely used to assess the robustness of causal conclusions in observational studies, yet its interaction with the structure of measured covariates is often overlooked. When latent confounders cannot be directly adjusted for and are instead controlled using proxy variables, strong associations between exposure and measured proxies can amplify sensitivity to residual confounding. We formalize this phenomenon in linear regression settings by showing that a simple ratio involving the exposure model coefficient and residual exposure variance provides an observable measure of this increased sensitivity. Applying our framework to smoking and lung cancer, we document how growing socioeconomic stratification in smoking behavior over time leads to heightened sensitivity to unmeasured confounding in more recent data. These results highlight the importance of multicollinearity when interpreting sensitivity analyses based on proxy adjustment.
\end{abstract}


\section{Introduction}

Confounding, especially unmeasured confounding, remains a central concern in observational studies. Few examples illustrate this more clearly than the long-standing debate over the relationship between smoking and lung cancer. The seminal paper by \citet{cornfield1959smoking} played a pivotal role in this discussion by examining how strong an unmeasured confounder would need to be to overturn a causal interpretation. Their analysis showed that any such factor would have to be implausibly powerful given the scientific and social context of the time. This conclusion was supported by a careful consideration of alternative explanations, including age, sex, marital status, diagnostic practices, urban residence, and environmental exposure, thereby providing compelling robustness for the causal link between smoking and lung cancer.

The success of sensitivity analysis is intricately tied to the quality of the measured covariates in an observational study \citep{hammond1964smoking}. Sensitivity analysis can strengthen the credibility of a well-designed study, but it cannot substitute for thoughtful covariate selection. The central role of measured confounders in the design of an observational study cannot be overstated and is pivotal in assessing the rigor of such an analysis. In fact, under reasonable conditions, measured covariates can even attenuate bias arising from unmeasured confounders \citep{zhang2025general}. In this regard, it is imperative to consider the effect of measured covariates on the sensitivity of causal conclusions.

In a correctly specified standard linear regression of an outcome $Y$ on an exposure $A$ and a measured confounder $X$, strong multicollinearity between $A$ and $X$ is traditionally viewed as a problem of statistical power \citep{farrar1967multicollinearity}. When variables are highly collinear, standard errors increase, and estimates become less precise, making it harder to detect treatment effects. 

A skeptic may argue that even under multicollinearity, effect estimates from a correctly specified linear model remain valid, and that observing a significant effect despite multicollinearity therefore strengthens causal evidence. However, this conclusion relies critically on the assumption of no unmeasured confounding. When such an assumption is questionable, the implications of strong multicollinearity between the exposure and measured confounders on the sensitivity of a causal claim remain unclear.

For instance, consider two studies as shown in Figure \ref{hypo_studies}, with an exposure $A$, a measured covariate $X$, and an outcome $Y$. Suppose $X$ has the same distribution in both studies and serves as a stable proxy for an unmeasured confounder $U$, as is common in observational studies. Figure \ref{hypo_studies} shows that the two studies report similar effect size estimates, standard errors, and sensitivity analyses (using the sensitivity model of \citet{cinelli2020making}). However, exposure is more strongly predicted by $X$ in Study 1 than in Study 2 (see Appendix B for the data-generating mechanisms). Analogous situations could arise if one instead used sensitivity models such as those in \citet{rosenbaum2010design} or \citet{tan2006distributional}. In this setting, does one study warrant greater confidence in its causal claims than the other?

\begin{figure}[!ht]
\centering
\label{hypo_studies}
\textbf{Covariate Exposure Relationship}
\includegraphics[width=0.99\linewidth]{exp_cov_relation.pdf}

\begin{minipage}{\linewidth}
  \centering
  \begin{minipage}[t]{0.48\linewidth}
    \textbf{Study 1 (R output)}\\[-0.3em]
    \begin{lstlisting}[style=rout]
Call:
lm(formula = y ~ a + x)

Residuals:
    Min      1Q  Median      3Q     Max 
 -4.236  -1.092  -0.031   1.018   4.769 

Coefficients:
             Estimate Std. Error t value Pr(>|t|)
(Intercept) -0.00160   0.04752   -0.034   0.9732
a            3.48138   0.05417   64.271  <2e-16**
x           -0.19174   0.09802   -1.956   0.0507

Residual standard error: 1.502 on 
    997 degrees of freedom
Multiple R-squared: 0.9523,  
Adjusted R-squared: 0.9522
F-statistic: 9953 on 2 and 997 DF,  
    p-value: < 2.2e-16

Sensitivity Analysis to Unobserved Confounding

Unadjusted Estimates of 'a':
  Coef. estimate: 3.48138
  Standard Error: 0.05417
  t-value: 64.27081

Sensitivity Statistics:
  Partial R2 of treatment with outcome: 0.80557
  Robustness Value (q = 1): 0.83266
  Robustness Value (q = 1, alpha = 0.05): 0.82515
    \end{lstlisting}
  \end{minipage}%
  \hspace{0.04\linewidth}%
  \begin{minipage}[t]{0.48\linewidth}
    \textbf{Study 2 (R output)}\\[-0.3em]
    \begin{lstlisting}[style=rout]
Call:
lm(formula = y ~ a + x)

Residuals:
    Min      1Q  Median      3Q     Max 
-4.0173 -0.9755 -0.0070  0.9448  3.8130 

Coefficients:
             Estimate Std. Error t value Pr(>|t|)
(Intercept) -0.01455   0.04312   -0.337   0.736
a            3.31667   0.05042   65.779  <2e-16**
x            1.42632   0.04153   34.348  <2e-16**

Residual standard error: 1.363 on 
    997 degrees of freedom
Multiple R-squared: 0.9024, 
Adjusted R-squared: 0.9022
F-statistic: 4609 on 2 and 997 DF,  
    p-value: < 2.2e-16

Sensitivity Analysis to Unobserved Confounding

Unadjusted Estimates of 'a':
  Coef. estimate: 3.31667
  Standard Error: 0.05042
  t-value: 65.7786

Sensitivity Statistics:
  Partial R2 of treatment with outcome: 0.81273
  Robustness Value (q = 1): 0.83813
  Robustness Value (q = 1, alpha = 0.05): 0.83096
    \end{lstlisting}
  \end{minipage}
\end{minipage}

\caption{Comparison of two hypothetical studies. (Top) Observed Exposure and Covariate Relationship (Bottom) Effect estimates and Sensitivity Analysis for Observed Outcome.}
\end{figure}

Observational studies in which the exposure is heavily confounded with other covariates have often been shown to yield misleading causal conclusions. Notable examples include early findings on the protective effects of hormone replacement therapy on coronary artery disease \citep{grodstein2001postmenopausal} and the adverse effects of caffeine consumption during pregnancy on birth weight \citep{vlajinac1997effect}, both of which were later overturned by experimental evidence. In each case, the exposure was strongly associated with measured differences in lifestyle, rendering the studies highly susceptible to unmeasured confounding. These examples serve as cautionary tales, underscoring the need for heightened skepticism when drawing causal conclusions from non-experimental data in settings where exposure is strongly confounded \citep{rutter2007identifying}.


In this article, we investigate this question further, studying how the relationship between measured covariates and exposure shapes the interpretation of sensitivity analyses in observational studies, particularly when these covariates proxy for latent confounders. Revisiting smoking as a motivating example, we show how evolving socioeconomic patterns in exposure uptake alter the sensitivity of causal claims over time. Together, these considerations offer a framework for a more nuanced interpretation of sensitivity analyses in the presence of strong exposure–covariate associations.

\section{Background}
\subsection{The measurement error problem in explanatory variables}
We begin by revisiting the standard linear regression setup where the explanatory variable is measured with error.
\begin{align}
    Y &= \beta_0 + \beta X^* + \varepsilon_Y; & X &= X^* + \varepsilon_X; & \varepsilon_Y\indep \varepsilon_X \label{model1}
\end{align}
It is well known that an ordinary least squares (OLS) regression coefficient $\beta_{Y\sim X}$ of $Y$ on $X$ in Model \eqref{model1} underestimates the true $\beta$, as given by Equation \eqref{atbias}.
\begin{equation}
    \beta_{Y\sim X} = \beta \dfrac{\Var(X^*)}{\Var(X^*) + \Var(\varepsilon_X)}. \label{atbias}
\end{equation}
This phenomenon is often referred to as attenuation bias in econometrics \citep{wooldridge2010econometric}, and it foreshadows the problem we address when there is strong multicollinearity between $A$ and $X$ in a regression, and when $X$ is itself associated with an unmeasured variable $U$. In such settings, the bias in estimating the slope of $A$ on $Y$ is affected by the association between $U$ and $X$, with direct implications for the sensitivity of the estimated coefficient to unmeasured confounding.

\subsection{Multicollinearity and the Role of Measured Covariates}
Consider the following structural equation
\begin{equation}
    Y = \beta_0 + \beta A + \theta_X X+ \varepsilon_Y, \ \\E[\varepsilon_Y|A,X] =0, \label{reg_model}
\end{equation}
where a linear regression of $Y$ on $A$ and $X$ identifies the structural parameter $\beta$. In this setting, the absence of unmeasured confounding implies that strong multicollinearity between the exposure $A$ and the measured covariates $X$ does not introduce bias, but instead affects only the precision of the estimated treatment effect by inflating its standard error. If the outcome is conditionally mean-independent of $X$ given $A$ (given by $\theta_X = 0$), including $X$ is unnecessary for unbiasedness and may reduce efficiency \citep{wooldridge2016should}. In practice, however, the assumption of no unobserved confounding is difficult to verify, and bias from omitted variables may persist despite adjustment using imperfect proxies.

A well-known illustration with substantial multicollinearity is the Equality of Educational Opportunity (Coleman) Report, which found that after controlling for students’ family background and proxies for socioeconomic status, differences in measured school resources like funding, teacher qualifications, and available facilities, had limited explanatory power for academic achievement, compared to family background and socioeconomic status \citep{coleman1968equality}. 

These conclusions have since been challenged by modern experimental and \sloppy{quasi-experimental} evidence showing that specific school inputs can have meaningful causal effects on achievement, including class-size reductions and teacher quality \citep{angrist1999using, chetty2011does}. One possible reconciliation of this paradox is that socioeconomic status is a latent construct that is only imperfectly captured by observed proxies such as parental education, family income, and home resources. When these proxies are highly correlated with school resources, strong multicollinearity combined with proxy measurement error can increase the sensitivity of causal estimates to residual unmeasured confounding, potentially exacerbating omitted-variable bias in observational analyses such as those in the Coleman Report. 

\subsection{The Smoking Example: Data from NHANES}

The issue of multicollinearity in measured covariates has been particularly salient in the evolution of smoking behavior over the last few decades. Socioeconomic gradients have long been associated with heterogeneity in smoking exposure, and the cigarette epidemic is commonly understood to have evolved through four stages \citep{lopez1994descriptive, mackenbach2006health}. \citet{mackenbach2006health} describes smoking as following a diffusion pattern across socioeconomic groups: early adoption occurs primarily among higher–socioeconomic-status men, but as smoking spreads more broadly, it becomes increasingly evenly distributed across the population, weakening socioeconomic gradients. In later stages, smoking prevalence declines first among more advantaged groups while remaining concentrated among less affluent populations, eventually becoming predominantly a habit of lower socioeconomic groups. These patterns have been empirically documented in subsequent studies \citep{pampel2005diffusion, vedoy2014tracing, quirmbach2016gender, di2019smoking, garrett2019socioeconomic}. 

To aid this discussion, we analyze data from the National Health and Nutrition Examination Survey (NHANES), a large-scale program designed to assess the health and nutritional status of adults and children in the United States \citep{nhanes}. To examine changes in smoking exposure over time, we focus on two survey periods for which NHANES data are available: 1971–1974 (NHANES I) and 2015–2016. We fit logistic propensity score models following \citet{rosenbaum1983central} for current smoking status in each period, adjusting for age at interview, race, highest grade completed, and poverty index, and stratifying by sex. Using NHANES I data, the model achieves a C-statistic of 0.61 for males and 0.65 for females; in contrast, the corresponding models using 2015–2016 NHANES data yield a C-statistic of approximately 0.70 for both sexes. This increase in predictive performance indicates that smoking status has become more strongly determined by observed demographic and socioeconomic covariates over time, suggesting that multicollinearity between smoking exposure and confounders is more pronounced today than it was fifty years ago.
\begin{table}[!ht]
\centering
\caption{Logistic regression results in both time periods, stratified by sex}
\label{tab:logit_sex_period}
\small
\setlength{\tabcolsep}{9pt}
\begin{threeparttable}
\begin{tabular}{lcccc}
\toprule
& \multicolumn{2}{c}{NHANES 1971-74} & \multicolumn{2}{c}{NHANES 2015-16} \\
\cmidrule(lr){2-3}\cmidrule(lr){4-5}
& Male & Female & Male & Female \\
\midrule
Intercept
& 2.063*** & 1.138*** & 2.449*** & 3.465*** \\
& (0.205)  & (0.175)  & (0.436)  & (0.566)  \\
\\
Age at interview
& -0.029*** & -0.037*** & -0.034*** & -0.039*** \\
& (0.003)   & (0.002)   & (0.005)   & (0.006)   \\
\\
Race: Black
& 0.174*    & 0.134     & 0.720***  & 0.063     \\
& (0.104)   & (0.082)   & (0.188)   & (0.209)   \\
\\
Race: Neither White nor Black
& 0.331     & -1.015**  & -0.371**  & -0.261    \\
& (0.347)   & (0.399)   & (0.153)   & (0.190)   \\
\\
Highest education grade
& -0.051*** & -0.034**  & -0.017    & -0.059*   \\
& (0.011)   & (0.011)   & (0.026)   & (0.034)   \\
\\
Poverty index
& -0.0004*  & 0.0007*** & -0.290*** & -0.315*** \\
& (0.0002)  & (0.0002)  & (0.049)   & (0.060)   \\
\bottomrule
\end{tabular}
\begin{tablenotes}[flushleft]
\footnotesize
\item Notes: Entries are log-odds coefficients from logistic regressions. Standard errors in parentheses.
\(\,^{***}p<0.01\), \(\,^{**}p<0.05\), \(\,^{*}p<0.10\).
\end{tablenotes}
\end{threeparttable}
\end{table}

Consistent with this observation, patterns in the logistic regression coefficients in Table \ref{tab:logit_sex_period} indicate that smoking is more tightly associated with socioeconomic and demographic characteristics in the contemporary data than in the early 1970s. Together, these findings motivate the \sloppy{real-data} counterpart to the conundrum posed in Figure \ref{hypo_studies}: are effect estimates of smoking on lung cancer today more sensitive to potential unmeasured confounding than they were in earlier periods?

\section{The Effect of Residual Confounding on Sensitivity}

\subsection{Implications on Omitted Variable Bias}
Consider an observed outcome $Y$, exposure of interest $A$, measured covariate $X$ and unmeasured covariate $U$, governed by the following structural equation \footnote{A structural equation specifies how a variable is generated as a function of other variables and an error term, encoding assumed causal relationships in the data-generating process \citep{goldberger1991course}.}. 
\begin{align}
    Y &= \alpha_0 + \beta A + \theta_X X+ \gamma U + \varepsilon_Y; \ \ \E[\varepsilon_Y|A,X, U] = 0. \label{ysem}
\end{align}
Furthermore, suppose the measured covariate $X$ is a proxy for $U$, given by 
\begin{equation} X = U + \varepsilon_X; \ \ \E[\varepsilon_X|U] = 0, \varepsilon_Y\indep \varepsilon _X. \label{xsem}\end{equation}
In addition, we also impose the following restriction.
\begin{assumption} \label{nocom}
    $ \text{Cov}(A,\varepsilon_X) = 0.  $
\end{assumption}
In simple words, Assumption \ref{nocom} is satisfied if $A$ and $X$ do not share any common cause except $U$. Note that, if a linear structural equation also holds for $A$, then Assumption \ref{nocom} also implies that $X$ is not a direct cause of $A$, and the only association between $A$ and $X$ occurs through $U$.   

\begin{proposition} \label{main}
    Suppose Equations \eqref{ysem} and \eqref{xsem} hold, and all random variables have finite second moments. Let $\beta_{Y\sim A,X}$ denote the coefficient of $A$ in the linear regression of $Y\sim A,X$. Then, under Assumption \ref{nocom}
    $$\beta_{Y\sim A,X} - \beta = \gamma\dfrac{\Var(\varepsilon_X)\beta_{A\sim X}}{\Var(A)(1-R^2_{A\sim X})},$$ where $\beta_{A\sim X}$ and $R^2_{A\sim X}$ are the coefficient of $X$ and the coefficient of determination respectively, in a linear regression of $A$ on $X$.
\end{proposition}

Proposition \ref{main}, the proof of which can be found in Appendix A, decomposes the bias into three distinct components: (i) $\gamma$, the confounding strength of $U$ on $Y$, (ii) $\Var(\varepsilon_X)$, which encapsulates the quality of $X$ as a proxy of $U$, and (iii) the observable strength of collinearity between $A$ and $X$, quantified by $\beta_{A\sim X}/(\Var(A)(1-R^2_{A\sim X}))$.
In particular, as the association between $X$ and $A$ strengthens, reflected in a larger effect size $\beta_{A\sim X}$, while the residual variance of $A$ (i.e., $\Var(A)(1-R^2_{A\sim X})$) remains stable, the discrepancy between the estimated coefficient $\beta_{Y\sim A,X}$ and the true causal parameter $\beta$ increases. This amplification heightens the sensitivity of causal conclusions to residual unmeasured confounding.

Notably, Proposition \ref{main} imposes no structural assumptions on the relationship between $A$ and $U$ beyond Assumption \ref{nocom}, and is agnostic to the functional form governing this association. The bias amplification it describes depends only on observable OLS regression coefficients of $A\sim X$, regardless of whether such a linear relationship holds. Furthermore, Proposition \ref{main} does not require $\theta_X$ to be null, and therefore its implications remain valid even when $X$ functions as a negative control outcome.\footnote{A negative control outcome is a variable that is not causally affected by the exposure of interest but may be associated with the underlying sources of bias.}

\subsection{Are smoking effect estimates today more sensitive to unmeasured confounding?}

Proposition \ref{main} offers a method to assess whether effect estimates of smoking on lung cancer from later data are more sensitive to unmeasured confounding. Let $Y$ denote lung-cancer incidence, $A$  smoking prevalence, $U$ the true unmeasured socioeconomic status and $X$ a measured proxy of $U$, such as the poverty index. We expect the structural equations governing $Y$ to be relatively stable across time periods, so that Equation \eqref{ysem} holds in both periods of our analysis. We make an analogous assumption for Equation \eqref{xsem}, arguing that the relationship between socioeconomic status and measured covariates such as race and poverty index has remained relatively stable over time. 


Note that the above conclusion is valid under Assumption \ref{nocom}, which excludes additional common causes between $A$ and $X$ beyond $U$. While violations of this assumption can complicate the relationship between $\beta_{Y\sim A,X}$ and $\beta$ (see Appendix A), the result remains most informative for covariates that plausibly act as proxies for latent socioeconomic status, such as the poverty index in the smoking application, and that are unlikely to directly affect smoking behavior.  With additional measured covariates, analogous conclusions hold conditional on those covariates.

\begin{table}[!ht]
\centering
\caption{Linear regression results in both time periods, stratified by sex}
\label{tab:lm_sex_period}
\small
\setlength{\tabcolsep}{9pt}
\begin{threeparttable}
\begin{tabular}{lcccc}
\toprule
& \multicolumn{2}{c}{NHANES 1971-74} & \multicolumn{2}{c}{NHANES 2015-16} \\
\cmidrule(lr){2-3}\cmidrule(lr){4-5}
& Male & Female & Male & Female \\
\midrule
Intercept
& 0.996*** & 0.715*** & 1.038*** & 1.269*** \\
& (0.048)  & (0.035)  & (0.093)  & (0.118)  \\
\\
Age at interview
& -0.0069*** & -0.0075*** & -0.0076*** & -0.0087*** \\
& (0.0006)   & (0.0004)   & (0.0010)   & (0.0013)   \\
\\
Race: Black
& 0.043*     & 0.026      & 0.160***   & 0.012      \\
& (0.025)    & (0.017)    & (0.041)    & (0.046)    \\
\\
Race: Neither White nor Black
& 0.079      & -0.195**   & -0.082**   & -0.060     \\
& (0.083)    & (0.070)    & (0.033)    & (0.042)    \\
\\
Highest education grade
& -0.012***  & -0.006***  & -0.003     & -0.013*   \\
& (0.0025)   & (0.0022)   & (0.0058)   & (0.0075)  \\
\\
Poverty index
& -0.00010*  & 0.00014*** & -0.064***  & -0.070*** \\
& (0.00005)  & (0.00004)  & (0.010)    & (0.013)   \\
\midrule\midrule
Residual variance $\hat\sigma^2$
& 0.239 & 0.206 & 0.219 & 0.220 \\
\bottomrule
\end{tabular}
\begin{tablenotes}[flushleft]
\footnotesize
\item Notes: Entries are OLS coefficients. Standard errors in parentheses.
\(\,^{***}p<0.01\), \(\,^{**}p<0.05\), \(\,^{*}p<0.10\).
\end{tablenotes}
\end{threeparttable}
\end{table}

Table \ref{tab:lm_sex_period} summarizes the linear regression coefficients for smoking uptake on measured covariates, including age at interview, race, highest education grade earned, and poverty index. The results indicate that the association between poverty index ($X$) and smoking uptake ($A$) has strengthened substantially over time, while the residual variance of smoking uptake has remained relatively stable. Under the assumption that poverty index does not directly affect smoking uptake except through latent socioeconomic status, we examine confidence intervals for the ratio of the poverty index coefficient to the residual variance across sexes and time periods as a diagnostic for differences in sensitivity to unmeasured confounding.

Conservative confidence intervals for the ratio $\beta_{A\sim X}/(\text{Var}(A)(1-R^2_{A\sim X}))$, controlling for all other measured covariates, are reported in Table \ref{bias_ratio}. The lack of overlap between these intervals across time periods, consistent across sexes, illustrates how the sensitivity of smoking–lung cancer effect estimates to unmeasured confounding has heightened over time. More pertinently, this exercise highlights the need for heightened caution in assessing the sensitivity of causal effects in settings where proxies for latent sources of bias become increasingly strongly associated with the exposure of interest.

\begin{table}[!htbp]
\centering
\caption{Confidence intervals for the ratio of the regression coefficient of the poverty index and regression residual variance of smoking uptake on measured covariates.}
\label{bias_ratio}
\small
\setlength{\tabcolsep}{10pt}
\begin{threeparttable}
\begin{tabular}{lcc}
\toprule
& Male & Female \\
\midrule
NHANES 1971-74
& $[-7.87\times10^{-4},\,1.82\times10^{-5}]$
& $[2.74\times10^{-4},\,1.12\times10^{-3}]$ \\
NHANES 2015-16
& $[-0.350,\,-0.215]$
& $[-0.387,\,-0.226]$ \\
\bottomrule
\end{tabular}
\begin{tablenotes}[flushleft]
\footnotesize
\item Notes: Entries are conservative 95\% confidence intervals, constructed by combining the Wald confidence interval for
${\beta}_{A\sim X}$ with the chi-square confidence interval for the residual variance in $A$, after controlling for other measured covariates.
\end{tablenotes}
\end{threeparttable}
\end{table}

\section{Discussion}

Sensitivity analysis serves as an important tool for assessing the credibility of causal claims, but it cannot substitute for careful consideration of measured covariates. While high association between treatment uptake and observed covariates has traditionally been viewed as a problem of statistical power, there has been limited attention on its implications for validity and sensitivity. In this article, we argue that when latent confounders are difficult to adjust for directly and are instead controlled for using proxies, stronger associations between the exposure and measured proxies imply effect estimates that are more prone to exacerbated omitted-variable bias.

We show that the ratio of the measured covariate coefficient in the exposure model to the residual variance of the exposure provides a natural, observable metric for evaluating this added sensitivity in linear regression based analyses. Applying this framework to smoking, we demonstrate that smoking exposure has become increasingly intertwined with socioeconomic factors over time, and the increased multicollinearity makes causal assertions from more recent data less robust to unmeasured confounding. We hope this work serves as a cautionary note when assessing the sensitivity of causal conclusions to potential unmeasured confounding, particularly when strong multicollinearity exists between exposure and covariates intended to proxy latent confounders.

\subsection*{Acknowledgment}{The authors are grateful to the participants of the Wharton Causal Data Science Lab for their insightful discussions and thoughtful feedback.}

\newpage

\appendix

\section*{Appendix A: Proof of Proposition \ref{main}} 
\label{app:theorem}
\proof

Let $(\beta_{Y\sim A, X}, \theta_{Y\sim A, X})$ denote the OLS regression coefficients when regressing $Y\sim A,X$. Under Model \eqref{ysem}, we have,
\begin{align*}
    \begin{bmatrix}
        \beta_{Y\sim A, X} \\ \theta_{Y\sim A,X}
    \end{bmatrix}
    &= \begin{bmatrix}
        \Var(A) & \text{Cov}(A,X) \\
        \text{Cov}(A,X) & \Var(X)
    \end{bmatrix}^{-1} \begin{bmatrix}
     \text{Cov}(A,Y) \\ \text{Cov}(X,Y)   
    \end{bmatrix} \\
    &= \dfrac{1}{\text{Var}(A)\Var(X) - \Cov(A,X)^2}\begin{bmatrix}
        \Var(X) & -\text{Cov}(A,X) \\
        -\text{Cov}(A,X) & \Var(A)
    \end{bmatrix} \begin{bmatrix}
     \text{Cov}(A,Y) \\ \text{Cov}(X,Y).   
    \end{bmatrix}
\end{align*}
Thus, 
\begin{equation} 
\beta_{Y\sim A, X} = \dfrac{\Var(X)\Cov(A,Y) - \Cov(A,X)\Cov(X,Y)}{\text{Var}(A)\Var(X) - \Cov(A,X)^2}. \label{betamid}
\end{equation}
Expanding on the right-hand side of Equation \eqref{betamid}, 
\begin{align*} 
& \Var(X)\Cov(A,Y) - \Cov(A,X)\Cov(X,Y) \\ 
&= \Var(X)[\beta \Var(A) + \theta_X\Cov(A,X) + \gamma \Cov(A,U)] \\
&\hspace{2em} - \Cov(A,X)[\beta \Cov(A,X) + \theta_X \Var(X) + \gamma \Cov(U,X)] \\
&= \beta [\Var(A)\Var(X) - \Cov(A,X)^2] + \gamma [\Var(X)\Cov(A,U) - \Cov(A,X)\Cov(U,X)].
\end{align*}
Thus, from Equation \eqref{betamid}, we have,
\begin{equation}
    \beta_{Y\sim A, X} - \beta = \gamma \dfrac{\Var(X)\Cov(A,U) - \Cov(A,X)\Cov(U,X)}{\Var(A)\Var(X) - \Cov(A,X)^2}. \label{diffmid}
\end{equation}
Continuing to expand on the numerator on the right-hand side of Equation \eqref{diffmid} yields 
\begin{align*}
    & \Var(X)\Cov(A,U) - \Cov(A,X)\Cov(U,X) \\
    &= (\Var(\varepsilon_X) + \Var(U))\Cov(A,U) - (\Cov(A,U) + \Cov(A,\varepsilon_X))\Var(U) \\
    &= \Var(\varepsilon_X)\Cov(A,U) - \Cov(A,\varepsilon_X)\Var(U) \\
    &= \Var(\varepsilon_X)[\Cov(A,X) - \Cov(A,\varepsilon_X)] - \Cov(A,\varepsilon_X)[\Var(X) - \Var(\varepsilon_X)]\\
    &= \Var(\varepsilon_X)\Cov(A,X) - \Cov(A,\varepsilon_X)\Var(X).
\end{align*}
Thus, a general form of the bias expression of $\beta_{A\sim X}$ from Equation \eqref{diffmid} is given by
\begin{align}
    \beta_{Y\sim A, X} - \beta &= \gamma \dfrac{\Var(\varepsilon_X)\Cov(A,X) - \Cov(A,\varepsilon_X)\Var(X)}{\Var(A)\Var(X) - \Cov(A,X)^2} \nonumber \\
    &= \gamma \dfrac{\Var(\varepsilon_X)\beta_{A\sim X} - \Cov(A,\varepsilon_X)}{\Var(A)(1-R^2_{A\sim X})}. \label{alfinal}
\end{align}
Note that, without Assumption \ref{nocom}, the direction of the omitted variable bias is unclear, with the precise direction depending not only on the observed measured strength of confounding $\beta_{A\sim X}$ but also on the strength of $X$'s direct effect on $A$ (not mediated by $U$) relative to the precision of $X$ as a proxy for $U$. Finally, under Assumption \ref{nocom}, Equation \eqref{alfinal} reduces to the expression in Proposition \ref{main}. \hfill $\blacksquare$ 





\section*{Appendix B: Data generation for Figure \ref{hypo_studies}}
The following data snippet produces the results in Figure \ref{hypo_studies}. Clearly, Study 1 has higher residual confounding, and thus produces more biased estimates than Study 2. 

\begin{verbatim}
library(sensemakr); library(ggplot2)
set.seed(2026)

u1 = rnorm(1000)
x1 = u1 + 0.5*rnorm(1000)
a1 = 2*u1 + 0.05*rnorm(1000)
y1 = 2.4*a1 + 2*u1 + 0*x1 + 1.5*rnorm(1000)

u2 = rnorm(1000)
x2 = u2 + 0.5*rnorm(1000)
a2 = 0.5*u2 + 0.8*rnorm(1000)
y2 = 3*a2 + 2*u2 + 0*x2 + rnorm(1000)

df1 <- data.frame(a = a1, x = x1, Study = "Study 1")
df2 <- data.frame(a = a2, x = x2, Study = "Study 2")
df  <- rbind(df1, df2)
xlim_all <- range(df$x); ylim_all <- range(df$a)

ggplot(df, aes(x = x, y = a)) +
  geom_point(alpha = 0.3) +
  geom_smooth(method = "lm", se = FALSE, color = "red") +
  facet_wrap(~ Study) +
  coord_cartesian(xlim = xlim_all, ylim = ylim_all) +
  labs(x = "Covariate x", y = "Exposure a") +
  geom_hline(yintercept = 0, linetype = "dashed") +
  theme_bw() +
  theme(strip.text = element_text(size = 12))

print_res = function(i){
  if(i== 1) {u = u1; x = x1; a = a1; y = y1}
  if(i== 2) {u = u2; x = x2; a = a2; y = y2}
  mod = lm(y ~ a + x); expmod = lm(a ~ x)
  print(summary(mod))
  sensemakr(mod, treatment = "a")
}

print_res(1)
print_res(2)
\end{verbatim}

\vskip 0.2in
\bibliography{sample}

@article{cornfield1959smoking,
  title={Smoking and lung cancer: recent evidence and a discussion of some questions},
  author={Cornfield, Jerome and Haenszel, William and Hammond, E Cuyler and Lilienfeld, Abraham M and Shimkin, Michael B and Wynder, Ernst L},
  journal={Journal of the National Cancer institute},
  volume={22},
  number={1},
  pages={173--203},
  year={1959},
  publisher={Oxford University Press}
}

@article{rosenbaum1983central,
  title={The central role of the propensity score in observational studies for causal effects},
  author={Rosenbaum, Paul R and Rubin, Donald B},
  journal={Biometrika},
  volume={70},
  number={1},
  pages={41--55},
  year={1983},
  publisher={Oxford University Press}
}

@article{lopez1994descriptive,
  title={A descriptive model of the cigarette epidemic in developed countries},
  author={Lopez, Alan D and Collishaw, Neil E and Piha, Tapani},
  journal={Tobacco control},
  volume={3},
  number={3},
  pages={242},
  year={1994}
}

@book{mackenbach2006health,
  title={Health inequalities: Europe in profile},
  author={Mackenbach, Johan P and others},
  year={2006},
  publisher={Produced by COI for the Department of Health}
}

@article{zhang2025general,
  title={A general condition for bias attenuation by a nondifferentially mismeasured confounder},
  author={Zhang, Jeffrey and Lee, Junu},
  journal={Biometrika},
  pages={asaf026},
  year={2025},
  publisher={Oxford University Press}
}

@article{garrett2019socioeconomic,
  title={Socioeconomic differences in cigarette smoking among sociodemographic groups},
  author={Garrett, Bridgette E and Martell, Brandi N and Caraballo, Ralph S and King, Brian A},
  journal={Preventing chronic disease},
  volume={16},
  pages={E74},
  year={2019}
}

@article{chetty2011does,
  title={How does your kindergarten classroom affect your earnings? Evidence from Project STAR},
  author={Chetty, Raj and Friedman, John N and Hilger, Nathaniel and Saez, Emmanuel and Schanzenbach, Diane Whitmore and Yagan, Danny},
  journal={The Quarterly journal of economics},
  volume={126},
  number={4},
  pages={1593--1660},
  year={2011},
  publisher={MIT Press}
}

@article{angrist1999using,
  title={Using Maimonides' rule to estimate the effect of class size on scholastic achievement},
  author={Angrist, Joshua D and Lavy, Victor},
  journal={The Quarterly journal of economics},
  volume={114},
  number={2},
  pages={533--575},
  year={1999},
  publisher={MIT Press}
}

@article{wooldridge2016should,
  title={Should instrumental variables be used as matching variables?},
  author={Wooldridge, Jeffrey M},
  journal={Research in Economics},
  volume={70},
  number={2},
  pages={232--237},
  year={2016},
  publisher={Elsevier}
}

@article{pampel2005diffusion,
  title={Diffusion, cohort change, and social patterns of smoking},
  author={Pampel, Fred C},
  journal={Social science research},
  volume={34},
  number={1},
  pages={117--139},
  year={2005},
  publisher={Elsevier}
}

@article{vedoy2014tracing,
  title={Tracing the cigarette epidemic: An age-period-cohort study of education, gender and smoking using a pseudo-panel approach},
  author={Ved{\o}y, Tord F},
  journal={Social Science Research},
  volume={48},
  pages={35--47},
  year={2014},
  publisher={Elsevier}
}

@article{quirmbach2016gender,
  title={Gender, education and Russia’s tobacco epidemic: A life-course approach},
  author={Quirmbach, Diana and Gerry, Christopher J},
  journal={Social Science \& Medicine},
  volume={160},
  pages={54--66},
  year={2016},
  publisher={Elsevier}
}

@article{di2019smoking,
  title={The smoking epidemic across generations, genders, and educational groups: A matter of diffusion of innovations},
  author={Di Novi, Cinzia and Marenzi, Anna},
  journal={Economics \& Human Biology},
  volume={33},
  pages={155--168},
  year={2019},
  publisher={Elsevier}
}

@article{coleman1968equality,
  title={The concept of equality of educational opportunity},
  author={Coleman, James},
  journal={Harvard educational review},
  volume={38},
  number={1},
  pages={7--22},
  year={1968},
  publisher={Harvard Education Publishing Group}
}

@book{wooldridge2010econometric,
  title={Econometric analysis of cross section and panel data},
  author={Wooldridge, Jeffrey M},
  year={2010},
  publisher={MIT press}
}

@article{farrar1967multicollinearity,
  title={Multicollinearity in regression analysis: the problem revisited},
  author={Farrar, Donald E and Glauber, Robert R},
  journal={The review of economic and statistics},
  pages={92--107},
  year={1967},
  publisher={JSTOR}
}

@misc{nhanes,
title={{National Health and Nutrition Examination Survey Data}},
author = {{Centers for Disease Control and Prevention (CDC)}},
howpublished = {National Center for Health Statistics (NCHS)},
institution  = {U.S. Department of Health and Human Services},
  address      = {Hyattsville, MD},
  year         = {1971-74, 2015-16},
  url          = {https://www.cdc.gov/nchs/nhanes/index.htm}
}

@article{hammond1964smoking,
  title={Smoking in relation to mortality and morbidity. Findings in first thirty-four months of follow-up in a prospective study started in 1959},
  author={Hammond, E Cuyler},
  journal={Journal of the National Cancer Institute},
  volume={32},
  number={5},
  pages={1161--1188},
  year={1964},
  publisher={Oxford University Press}
}

@article{tan2006distributional,
  title={A distributional approach for causal inference using propensity scores},
  author={Tan, Zhiqiang},
  journal={Journal of the American Statistical Association},
  volume={101},
  number={476},
  pages={1619--1637},
  year={2006},
  publisher={Taylor \& Francis}
}

@article{cinelli2020making,
  title={Making sense of sensitivity: Extending omitted variable bias},
  author={Cinelli, Carlos and Hazlett, Chad},
  journal={Journal of the Royal Statistical Society Series B: Statistical Methodology},
  volume={82},
  number={1},
  pages={39--67},
  year={2020},
  publisher={Oxford University Press}
}

@book{goldberger1991course,
  title={A course in econometrics},
  author={Goldberger, Arthur Stanley},
  year={1991},
  publisher={Harvard University Press}
}

@article{grodstein2001postmenopausal,
  title={Postmenopausal hormone use and secondary prevention of coronary events in the Nurses' Health Study: a prospective, observational study},
  author={Grodstein, Francine and Manson, JoAnn E and Stampfer, Meir J},
  journal={Annals of internal medicine},
  volume={135},
  number={1},
  pages={1--8},
  year={2001},
  publisher={American College of Physicians}
}

@article{vlajinac1997effect,
  title={Effect of caffeine intake during pregnancy on birth weight},
  author={Vlajinac, Hristina D and Petrovi{\'c}, Radmila R and Marinkovi{\'c}, Jelena M and {\v{S}}ipeti{\'c}, Sandra B and Adanja, Benko J},
  journal={American journal of epidemiology},
  volume={145},
  number={4},
  pages={335--338},
  year={1997},
  publisher={Oxford University Press}
}

@book{rutter2007identifying,
  title={Identifying the Environmental Causes of Disease: How Should We Decide what to Believe and when to Take Action?: Report Synopsis},
  author={Rutter, Michael},
  year={2007},
  publisher={Academy of Medical Sciences}
}

@book{rosenbaum2010design,
  title={Design of observational studies},
  author={Rosenbaum, Paul R and Rosenbaum, P and Briskman},
  volume={10},
  year={2010},
  publisher={Springer}
}
\end{document}